\documentclass[letterpaper, 10 pt, conference]{ieeeconf}  

\IEEEoverridecommandlockouts                              
\let\labelindent\relax

\usepackage{grffile}
\usepackage{amsmath,amssymb,amsbsy}
\usepackage{hyperref}
\usepackage{amsthm} 
\usepackage{color}
\usepackage{comment}
\usepackage{lipsum}
\usepackage{graphicx}%
\usepackage{subfigure}
\usepackage[normalem]{ulem}
\usepackage{enumerate}
\usepackage{epstopdf}
\usepackage{bbm}
\newtheorem{theorem}{Theorem}[section]
\newtheorem{lemma}[theorem]{Lemma}
\newtheorem{definition}[theorem]{Definition}	
\newtheorem{proposition}[theorem]{Proposition}
	
\newtheorem{remark}[theorem]{Remark}	
\newtheorem{problem}[theorem]{Problem}
\newtheorem{example}[theorem]{Example}

\title{\LARGE \bf
Mean-Field Controllability and Decentralized Stabilization of Markov Chains, Part I: Global Controllability and Rational Feedbacks
}

\author{Karthik Elamvazhuthi, Vaibhav Deshmukh, Matthias Kawski, and Spring Berman
\thanks{This work was supported by National Science Foundation (NSF) Award CMMI-1436960 and by
ONR Young Investigator Award N00014-16-1-2605.}
\thanks{Karthik Elamvazhuthi, Vaibhav Deshmukh, and Spring Berman are with the School for Engineering of
Matter, Transport and Energy, Arizona State University, Tempe, AZ, 85281
USA {\tt\small \{karthikevaz, vdeshmuk, Spring.Berman\}@asu.edu}.}
\thanks{Matthias Kawski is with the School of Mathematical and Statistical Sciences,
Arizona State University, Tempe, AZ, 85281 USA {\tt\small \{kawski@asu.edu\}}.%
}}

\begin{document}

\maketitle
\thispagestyle{empty}
\pagestyle{empty}

\begin{abstract}
In this paper, we study the  controllability  and stabilizability properties of the Kolmogorov forward equation of a continuous time Markov chain (CTMC) evolving on a finite state space, using the transition rates as the control parameters. Firstly, we prove small-time local and global controllability from and to strictly positive equilibrium configurations when the underlying graph is strongly connected. Secondly, we show that there always exists a locally exponentially stabilizing decentralized linear (density-)feedback law that takes zero value at equilibrium and respects the graph structure, provided that the transition rates are allowed to be negative and the desired target density lies in the interior of the set of probability densities. For bidirected graphs, that is, graphs where a directed edge in one direction implies an edge in the opposite direction, we show that this linear control law can be realized using a decentralized rational feedback law of the form $k(\mathbf{x}) = a(\mathbf{x})+ b(\mathbf{x})\frac{f(\mathbf{x})}{g(\mathbf{x})}$ that also respects the graph structure and  control constraints (positivity and zero at equilibrium). This enables the possibility of using Linear Matrix Inequality (LMI) based tools to algorithmically construct decentralized density feedback controllers for stabilization of a robotic swarm to a target task distribution with no task-switching at equilibrium, as we demonstrate with several numerical examples.
\end{abstract}


\section{INTRODUCTION}
\label{sec:intro}

In recent years, there has been considerable work on approaches to task allocation for a large number of homogeneous robots that switch stochastically between tasks at tunable transition rates \cite{berman2009optimized,martinoli2004modeling,mather2014synthesis}. In these approaches, the robots' states evolve according to a continuous time Markov chain, and their task distribution is controlled using the corresponding mean-field model. This method enables scalable control design due to independence of the control methodology from agent numbers.  It has many applications in robotics, such as environmental monitoring, surveillance, disaster response, and autonomous construction.

Multiple approaches have been proposed in the literature for control synthesis in this framework. Optimal stabilization of the Kolmogorov forward equation using time-invariant constant inputs was considered in \cite{berman2009optimized}. Optimal control using time-varying control parameters has been addressed in several different contexts such as control of swarms \cite{bandyopadhyay2016probabilistic,demir2015decentralized}, mean-field games \cite{gomes2013continuous}, and optimal transport \cite{solomon2016continuous}. To improve convergence rates to the stationary distribution, these control approaches have also been extended to the case where the density of the swarm is fed back to the agents \cite{demir2015decentralized,mather2014synthesis}. Since density feedback requires global information, these works have also considered decentralized control approaches either by {\it a priori} restricting the controller to have a decentralized structure \cite{mather2014synthesis} or by designing a centralized controller and then using estimation algorithms to estimate the global density of the swarm in a decentralized manner \cite{demir2015decentralized}. Here, by {\it decentralized} we mean that each agent's controller or estimation parameters depend only on information that the agent can obtain from its local environment.

In this paper, we make two contributions to the mean-field control problem. First, we study local and global controllability properties of the forward equation when the control inputs are required to be zero at equilibrium. The case when control inputs are not constrained to be zero  at equilibrium is comparatively much easier, since local controllability follows directly from linearization based arguments, so we do not consider this case here. Second, we address the stabilization of mean-field models using decentralized feedback under the constraint that the transition rates are required to be zero at equilibrium. Such a constraint is needed in swarm robotic applications to prevent robots from constantly switching between states at equilibrium. We have shown that when this constraint is not imposed, a large class of target distributions (target densities with strongly connected supports) can be stabilized even without any density feedback \cite{biswal2016sos}.

The problem of unnecessary task-switching at equilibrium was previously addressed for CTMCs in \cite{mather2014synthesis} as a variance control problem, and for DTMCs in \cite{bandyopadhyay2016probabilistic} using a decentralized density estimation strategy that implements centralized  feedback laws and ensures that the transition matrix is the identity matrix at equilibrium. In this paper, we investigate the CTMC case in more detail. In contrast to \cite{mather2014synthesis}, we explicitly show that any (strictly positive) distribution is stabilizable using a decentralized feedback law, and we impose the additional constraint that transition rates must be zero at equilibrium. Moreover, the controller in \cite{mather2014synthesis} was proven to be stabilizing under the assumption that negative transition rates are admissible, and was then implemented with a saturation condition in order to avoid negative rates, in which case the stability guarantees are lost. We show how this issue can be resolved by interpreting a negative flow from one state to another as a positive flow of appropriate magnitude in the opposite direction.

\section{NOTATION}

We first define the notation that will be used to formulate the problems addressed in this paper. We denote by $\mathcal{G} = (\mathcal{V},\mathcal{E})$ a directed graph with a set of $M$ vertices, $\mathcal{V} = \lbrace 1,2,...,M \rbrace$, and a set of $N_{\mathcal{E}}$ edges, $  \mathcal{E} \subset \mathcal{V} \times \mathcal{V}$. An edge from vertex $i \in \mathcal{V}$ to vertex $j \in \mathcal{V}$ is denoted by $e = (i,j) \in \mathcal{E}$. We define a source map $S : \mathcal{E} \rightarrow \mathcal{V}$  and a target map  $T: \mathcal{E} \rightarrow \mathcal{V}$ for which $S(e) = i$ and $T(e) = j$ whenever $e = (i,j) \in \mathcal{E}$. There is a {\it directed path} of length $s$ from node $i\in \mathcal{V}$ to node $j\in \mathcal{V}$ if there exists a sequence of edges $\{e_i\}^s_{i=1}$ in $\mathcal{E}$ such that $S(e_1) = i$, $T(e_s) = j$, and $S(e_k)=T(e_{k-1})$ for all $1\leq k < s-1$. A directed graph $\mathcal{G}=(\mathcal{V},\mathcal{E})$ is called {\em strongly connected} if for every pair of distinct vertices $v_0,\,v_T\in \mathcal{V}$, there exists a {\em directed path} of edges in $\mathcal{E}$ connecting $v_0$ to $v_T$. We assume that $(i,i) \notin \mathcal{E}$ for all $i \in \mathcal{V}$. The graph  $\mathcal{G}$ is said to be {\it bidirected} if $e \in \mathcal{E}$ implies that $\tilde{e} = (T(e), S(e))$ also lies in $\mathcal{E}$.

We denote the $M$-dimensional Euclidean space by $\mathbb{R}^M$. $\mathbb{R}^{M \times N}$ will refer to the space of $M \times N$ matrices, and $\mathbb{R}_+$ will refer to the set of positive real numbers. Given a vector $\mathbf{x}  \in \mathbb{R}^M$, $x_i$ will refer to the $i^{th}$ coordinate value of $\mathbf{x}$. The $2-$norm of the vector $\mathbf{x}  \in \mathbb{R}^M$ is denoted by $\|\mathbf{x}\|_2 = \sqrt{\sum_i x^2_i}$. For a matrix $\mathbf{A} \in \mathbb{R}^{M \times N}$, $A^{ij} $ will refer to the element in the $i^{th}$ row and $j^{th}$ column of $\mathbf{A}$. For a subset $B \subset \mathbb{R}^M$, ${\rm int}(B)$ will refer to the interior of the set $B$.

\section{PROBLEM FORMULATION}

We consider a population of $N$ autonomous agents that must reallocate among a set of states, such as tasks that must be performed in different spatial regions, to achieve a target population distribution at equilibrium. Each agent has a finite state space $\mathcal{V}$, where each vertex in $\mathcal{V}$ represents a different state. The edges in $\mathcal{E}$ define the possible agent transitions between vertices. Denoting the set of admissible control inputs by $U \subset \mathbb{R}$, the agents' transition rules are determined by the control parameters $u_{e}:[0,\infty) \rightarrow U$ for each $e \in \mathcal{E}$, also known as the {\it transition rates} of the associated CTMC. An agent at state $i$ at time $t$ decides to switch to state $j$ at probability per unit time $u_{e}(t)$, $e = (i,j)$. We focus on the case where $U \subset \mathbb{R}_+$, since transition rates must always be positive for a CTMC.

The state of each agent $i \in \{1,...,N\}$ is defined by a stochastic process $X_i(t)$ that evolves on the state space $\mathcal{V}$ according to the conditional probabilities 
\begin{equation}
\label{eq:marko}
\mathbb{P}\left(X_i(t+h) = T(e) | X_i(t) = S(e)\right) =~ u_{e}(t)h + o(h)
\end{equation}
for each $e \in \mathcal{E}$.  Here, $o(h)$ is the little-oh symbol and $\mathbb{P}$ is the underlying probability measure induced on the space of events $\Omega$ (which will be left undefined, as is common) by the stochastic processes $\lbrace X_i(t) \rbrace_{i=1}^N $. Let  $\mathcal{P}(\mathcal{V}) = \lbrace \mathbf{y} \in \mathbb{R}^M_+; ~\sum_v y_v = 1 \rbrace$ be the simplex of pro

bability densities on $\mathcal{V}$. Corresponding to the CTMC is a set of ordinary differential equations (ODEs) which determines the time evolution of the probability densities $\mathbb{P}(X_i(t) = v) = x_v(t) \in \mathbb{R}_+$. Since $\lbrace X_i \rbrace_{i=1}^N $ is a set of independent and identically distributed random variables, the {\it Kolmogorov forward equation} can be represented by a single linear system of ODEs,
\begin{eqnarray}
\label{eq:ctrsys}
\dot{\mathbf{x}}(t) &=& \sum_{ e\in \mathcal E} u_e(t) \mathbf{B}_e \mathbf{x}(t), \hspace{3mm} t \in [0, \infty), \\ \nonumber
\mathbf{x}(0) &=& \mathbf{x}^0 \in \mathcal{P}(\mathcal{V}),
\end{eqnarray}
where $\mathbf{B}_e$ are control matrices whose entries are given by
\[
 B_e^{ij} =
  \begin{cases}
   -1 & \text{if } i = j=  S(e),\\
    1 & \text{if } i= T(e), \hspace{1mm} j = S(e),\\
   0       & \text{otherwise.}
  \end{cases}
\]

The focus of this paper is to study controllability and stabilizability properties of the control system \eqref{eq:ctrsys}. To describe the controllability problem of interest, we first recall some controllability notions from nonlinear control theory \cite{bloch2015nonholonomic}.
\begin{definition}
Given $U \subset \mathbb{R}$ and $\mathbf{x}^0  \in \mathcal{P}(\mathcal{V})$, we define $R^{U}(\mathbf{x}^0,t)$ to be the set of all $\mathbf{y} \in  \mathcal{P}(\mathcal{V})$ for which there exists an admissible control, $\mathbf{u} = \lbrace u_e \rbrace_{e \in \mathcal{E}}$, taking values in $U$ such that there exists a trajectory of  system (\ref{eq:ctrsys}) with $\mathbf{x}(0) = \mathbf{x}^0$, $\mathbf{x}(t) = \mathbf{y}$. The \textbf{reachable set from $\mathbf{x}_0$ at time $T$} is defined to be
\begin{equation}
R_T^U(\mathbf{x}^0) = \cup_{0 \leq t \leq T} R^U(\mathbf{x}^0,t).
\end{equation}
\end{definition}

\begin{definition}
\label{def1}
The system \eqref{eq:ctrsys} is said to be {\bf small-time locally controllable (STLC)} from an equilibrium configuration $\mathbf{x}^{eq} \in \mathcal{P}(\mathcal{V})$ if the set of reachable states $R_T^U(\mathbf{x}^{eq})$ contains a neighborhood of $\mathbf{x}^{eq} \in \mathcal{P}(\mathcal{V})$ in the subspace topology of $\mathcal{P}(\mathcal{V})$ (as a subset of $\mathbb{R}^M$) for any $T>0$.
\end{definition}
Here, we have defined local controllability in terms of the subspace topology of $\mathcal{P}(\mathcal{V})$.  This is because the set $\mathcal{P}(\mathcal{V})$ is invariant for the system \eqref{eq:ctrsys} of controlled ODEs, and hence one cannot expect controllability to a full neighborhood of $\mathbf{x}^{eq}$. Informally, this just means that, due to conservation of mass, one cannot create or destroy agents by manipulating their rates of transitioning from one vertex to another.

Our first problem of interest can be framed as follows:

\begin{problem}
Given $\mathbf{x}^{eq} \in \mathcal{P}(\mathcal{V})$, determine if the system  \eqref{eq:ctrsys} is STLC from $\mathbf{x}^{eq} \in \mathcal{P}(\mathcal{V})$.
\end{problem}

Next, we consider the feedback stabilization problem for system \eqref{eq:ctrsys}. Consider the following system:
\begin{eqnarray}
\label{eq:fctrsys}
\dot{\mathbf{x}}(t)&=& \sum_{ e\in \mathcal E} k_e(\mathbf{x}) \mathbf{B}_e \mathbf{x}(t), \hspace{3mm} t \in [0, \infty), \\ \nonumber
\mathbf{x}(0) &=& \mathbf{x}^0 \in \mathcal{P}(\mathcal{V}).
\end{eqnarray}
\begin{problem}

Given $\mathbf{x}^{eq} \in \mathcal{P}(\mathcal{V})$, determine whether there exists a decentralized feedback law, defined as a collection of maps $\tilde{k}_e: \mathbb{R}^2 \rightarrow \mathbb{R}_+$ where $k_e(\mathbf{y}) = \tilde{k}_e(y_{S(e)},y_{T(e)})$ for each $e \in \mathcal{E}$ and $\mathbf{y} \in \mathbb{R}^M$,
such that for the closed-loop system \eqref{eq:fctrsys}, $\mathbf{x}^{eq}$ is asymptotically stable and $k_e(\mathbf{x}^{eq})=0$ for each $e \in \mathcal{E}$.
\end{problem}

Due to the dependence of the feedback control law $\lbrace k_e \rbrace_{e \in \mathcal{E}}$ on the probability densities $x_v$, the independence of the stochastic processes $\lbrace X_i(t) \rbrace_{i=1}^N$ is lost. Hence, for a finite number of agents, the time evolution of the probability densities $x_v$ cannot be described by a system of ODEs on $\mathbb{R}^M$ such as \eqref{eq:ctrsys}.  However, system \eqref{eq:fctrsys}
represents the evolution of the probability densities in the sense of the {\it mean-field hypothesis}.  That is, we take the limit $N \rightarrow \infty$ to obtain the population density $x_v(t) = \lim_{N \rightarrow \infty} \sum_{i=1}^N \frac{ \mathbbm{1}_v(X_i(t))}{N}$ for each $v\in \mathcal{V}$, where $ \mathbbm{1}_v:\mathcal{V} \rightarrow \lbrace 0,1 \rbrace$ is the indicator function of $v$. See \cite{kolokoltsov2011markov}[Chapter 5] for more details.

\section{ANALYSIS}

\subsection{Controllability}

In this section, we investigate the controllability properties of the system \eqref{eq:ctrsys}.

\begin{proposition}
\label{MK1}
If the graph $\mathcal{G}=(\mathcal{V},\mathcal{E})$ is not strongly
connected, then the system \eqref{eq:ctrsys} is not locally controllable.
\end{proposition}

\begin{proof}
Suppose that $\mathcal{G}=(\mathcal{V},\mathcal{E})$ is not strongly
connected. Then there exist vertices $v_1,\,v_2\in \mathcal{V}$ such
that there does not exist a path in $\mathcal{E}$ from $v_2$ to $v_1$.
Let $\mathcal{V}_1$ be the subset of vertices $v\in \mathcal{V}$ such
that $v=v_1$ or there exists a path in $\mathcal{E}$
from $v$ to $v_1$. Analogously,
let $\mathcal{V}_2$ be the subset of vertices $v\in \mathcal{V}$ such
that $v=v_2$ or there exists a path in $\mathcal{E}$
from $v_2$ to $v$.
Since there does not exist a path in $\mathcal{E}$ from $v_2$ to $v_1$,
it is clear that $\mathcal{V}_1$ and $\mathcal{V}_2$ are disjoint
and both are nonempty.
Then the output function $\varphi\colon \mathcal{P}(\mathcal{V})\mapsto \mathbb{R}$
defined by
\begin{equation}
\varphi (x)= \sum_{v\in \mathcal{V}_2}x_v- \sum_{v\in \mathcal{V}_1}x_v
\end{equation}
is nondecreasing along every solution curve of the
system~\eqref{eq:ctrsys}, which therefore is not locally controllable.
\end{proof}

\begin{proposition}
\label{MK2}
If the graph $\mathcal{G}=(\mathcal{V},\mathcal{E})$ is strongly
connected, then the system \eqref{eq:ctrsys} is STLC from every point in ${\rm int}(\mathcal{P}(\mathcal{V}))$.
\end{proposition}

Before proving the proposition, it is helpful to take a closer
look at the relations in the Lie algebra of the control vector fields,
and the corresponding product on the semi-group generated
by their exponentials.

Suppose that $e=(i,j),e'=(k,\ell) \in \mathcal{E}$ are two edges.
If $\{i,j\},\{k,\ell\}\subseteq \mathcal{V}$ are disjoint,
then the control matrices $\mathbf{B}_{(i,j)}$ and $\mathbf{B}_{(k,\ell)}$
commute, and hence so do their exponentials.
If $k=j$ and $\ell\neq j$, then $\mathbf{B}_{(i,j)} \mathbf{B}_{(j,\ell)}=\mathbf{0}$,
and the commutator evaluates to
\begin{equation}
[\mathbf{B}_{(j,\ell)},\mathbf{B}_{(i,j)}]=\mathbf{B}_{(j,\ell)} \mathbf{B}_{(i,j)}=\mathbf{B}_{(i,j)}-\mathbf{B}_{(i,\ell)}.
\end{equation}
From this, we can conclude that if the graph
is strongly connected, then the Lie algebra spanned
by the control vector fields $f_e\colon \mathbf{x}\mapsto \mathbf{B}_e \mathbf{x}$
spans the tangent space $T_\mathbf{x} (\mathcal{P}(\mathcal{V}))$ at every point $\mathbf{x}\in {\rm int}(\mathcal{P}(\mathcal{V}))$. However, since in our case $\mathbf{0}$ is not an interior point
of the convex hull of admissible control values $\mathbf{u}\in [0,\infty)^M$,
classical results on STLC do not apply directly.

For any edge $e=(i,j)\in \mathcal{E}$, the exponential
of the control matrix $\mathbf{B}_e$ is a stochastic matrix
with entries given by
\begin{equation}
(\exp t\mathbf{B}_e)_{k\ell} = \left\{ \begin{array}{cl}
1&\mbox{ if } k=\ell\neq S(e) \\
e^{-t}&\mbox{ if } k=\ell=S(e) \\
1-e^{-t}&\mbox{ if } k=T(e) \mbox{ and } \ell=S(e) \\
0&\mbox{ otherwise.}
\end{array} \right.
\end{equation}

Rather than writing out a general formula for the corresponding
product on the group for general edges $(i,j),(j,\ell)\in \mathcal{E}$,
we only state the product for the special case of
$\mathcal{V}=\{1,2,3\}$ and edges $e=(1,2)$ and $e'=(2,3)$: 
\begin{equation}
e^{t\mathbf{B}_{e'}}e^{s\mathbf{B}_e}\!=\!
\left( \begin{array}{ccc}
e^{-s} & 0 & 0\! \\
e^{-t}(1-e^{-s}) & e^{-t} & 0\! \\
(1-e^{-t})(1-e^{-s}) & 1-e^{-t} & 1\!
\end{array} \right).
\end{equation}

\begin{proof} (of Proposition \ref{MK2}).
Suppose that the graph $\mathcal{G}=(\mathcal{V},\mathcal{E})$
is strongly connected.
Fix an arbitrary point $\mathbf{x}^0\in {\rm int}(\mathcal{P}(\mathcal{V}))$.
Then there exists $\rho >0$ such that each coordinate
$x^0_i>2\rho$.
Let $\Delta \mathbf{x}\in [-\rho/M ,\rho/M ]^M$ be arbitrary but fixed such
that $\sum_{v\in \mathcal{V}} \Delta x_v =0$.
Let the final time $T>0$ be arbitrary but fixed.
We explicitly construct a piecewise constant control
$\mathbf{u} \colon [0,T] \mapsto [0,\infty)^{N_\mathcal{E}}$ that
steers the system \eqref{eq:ctrsys} from $\mathbf{x}^0$ at time $0$
to $\mathbf{x}^0+\Delta \mathbf{x}$ at time $T$.

Let $v_0 \in \mathcal{V}$ be arbitrary but fixed.
As a consequence of Proposition \ref{MK1}, there exists a path $\gamma=(e_1,\ldots, e_s)$ of edges in $\mathcal{E}$
that connects $v_0=S(e_1)$ back to $T(e_s)=v_0$ and which
{\em visits} every vertex $v\in \mathcal{V}$ at least once.
For $1\leq i\leq s$, let $v_i=T(e_i)$.
Let $\Delta t = T/s$.
Define the finite sequence $\{\delta_i\}_{i=1}^s\in \{0,1\}^{s}$
by $\delta_i=1$ if for all $i<j<s$, $v_j\neq v_i$, i.e.
the edge $e_i\in \gamma$ is the last edge whose source is
$S(e_i)=v_i$.
This sequence ensures that a control variation in the direction
of $x_v$ is only taken along the last edge that starts at $v$.
Finally, define a finite sequence $\{\sigma_i\}_{i=0}^s\in [-\rho,\rho]^{s}$
that keeps track of the accumulated control variations, where
\begin{equation}
\label{sigmadef}
\sigma_0=0, ~~~~\sigma_i = \sum_{j=1}^{i} \delta_j\Delta x_{v_{j-1}}, ~ 1<i\leq s.
\end{equation}
Note that if the path $\gamma$ is a Hamiltonian cycle, then
$s=M$, $\delta_i = 1$ for all $i=1,...,s$, 
and $\sigma _i= \sum_{j=1}^{i} \Delta x_{v_{j-1}}$,
which simplifies the formula \eqref{coron} below.

To distinguish between the two cases where $S(e_i)=v_i=v_0$ and $S(e_i)=v_i\neq v_0$, we introduce the vector $\mathbf{y}^0 \in \mathbb{R}^M$ by setting
$y^0_{v_0}=x_{v_0}-\rho$ and $y^0_{v_i}=x_{v_i}$ if $v_i\neq v_0$.
Consider the piecewise constant control
$\mathbf{u} \colon [0,T] \mapsto [0,\infty)^{N_\mathcal{E}}$ that is defined
on each interval $t\in [i\Delta t,(i+1)\Delta t)$, $0\leq i <s$, as 
\begin{equation}
\label{coron}
u_{e_i}(t)=
-\frac{1}{\Delta t}\log \left( 1-\frac{\rho-\sigma_i}{y^0_{v_i}+\rho-\sigma_{i-1}}
\right) 
\end{equation}
and $u_e(t) \equiv 0$ for all $e\neq e_i$. 

The key idea in this construction is that the much simpler
control obtained by setting $\Delta \mathbf{x}=\mathbf{0}$ in definition \eqref{sigmadef} successfully moves
a mass $\rho>0$ from $v_0$ along the path $\gamma$ and
back to $v_0$.
It is critical that the 
component $u_{e_i}$ of the control be strictly positive on the $i^{th}$ interval of time,
which enables the application of classical signed control variations to this component on the interval.
The explicit introduction of the nonzero $\Delta \mathbf{x}$ then
allows the following endpoint map to be solved explicitly:  
\begin{equation}
\stackrel{\longleftarrow}{\prod_{1\leq i\leq s}}
\exp \left(\Delta t u_{e_i} \mathbf{B}_{e_i} \right) \cdot \mathbf{x}^0= \mathbf{x}^0+\Delta  \mathbf{x},
\end{equation}
resulting in Equation \eqref{coron} for the control.
\end{proof}


Since the size of the achievable  $\Delta \mathbf{x}$ is bounded by half the
distance $\rho$ of the starting point from the boundary of
the simplex, and the control set does not contain $\mathbf{u}=\mathbf{0}$ in its interior,
the sizes of the small-time reachable sets are not immediately apparent.
However, the following holds:

\begin{theorem}
\label{MK0}
If the graph
$\mathcal{G}=(\mathcal{V},\mathcal{E})$ is strongly
connected, then the system \eqref{eq:ctrsys} is small-time globally
controllable from every point in the interior of the simplex $\mathcal{P}(\mathcal{V})$.
\end{theorem}

\begin{proof}
Suppose that $T>0$ and $\mathbf{x}^0, \mathbf{x}^T\in {\rm int}(\mathcal{P}(\mathcal{V}))$.
Let $\rho=\frac{1}{2} \min \{  \mathbf{x}^0_v, \mathbf{x}^T_v \colon v\in \mathcal{V}\}$,
$L=\| \mathbf{x}^T- \mathbf{x}^0\|_1$, and $N={\rm ceil}(L/\rho)$.
Partition the straight-line segment from $\mathbf{x}^0$ to $\mathbf{x}^T$ into $N$ segments,
e.g. with endpoints $\mathbf{y}^k= \mathbf{x}^0+\frac{k}{N}(\mathbf{x}^T- \mathbf{x}^0)$ for $0\leq k \leq N$.
Using Proposition \ref{MK2}, there exist controls $u^k\colon [\frac{kT}{N},\frac{(k+1)T}{N}]\mapsto [0,\infty)^{N_\mathcal{E}}$
that successively steer the system from $\mathbf{x}^k$ to $\mathbf{x}^{k+1}$.
Thus, the concatenation of these controls steers the system
from $\mathbf{x}^0$ to $\mathbf{x}^T$ in time $T$ using piecewise constant
controls that take values only on the axes of $(\mathbb{R}_+)^{M}$.
\end{proof}

It would be desirable to extend the above result to target distributions that lie on the boundary of $\mathcal{P}(\mathcal{V})$, for which agent population densities in some states are zero at equilibrium.  However, as we demonstrate in the following example, one cannot expect to reach target distributions on the boundary in finite time. The boundary points of $\mathcal{P}(\mathcal{V})$ are unreachable if the system starts from ${\rm int}(\mathcal{P}(\mathcal{V}))$, even if one uses possibly unbounded but measurable inputs with finite Lebesgue integrals.
\begin{example}
Consider the forward equation for a two-vertex bidirected graph,
\begin{eqnarray}
\dot{x}_1(t) &=&  -u_{(1,2)}(t)x_1(t)+u_{(2,1)}x_2(t), \label{eq:twovertex} \\
\dot{x}_2(t) &=&  u_{(1,2)}(t)x_1(t)-u_{(2,1)}x_2(t), \nonumber \\
&& \hspace{-6mm} x_1(0) = x^0_{1}, ~~x_2(0) = x^0_{2}. \nonumber
\end{eqnarray}
Let $u_{(1,2)},u_{(2,1)} \in L^1_+(0,1)$, the set of positive-valued measurable inputs with finite integrals over the time interval $(0,1)$. Then the solution, $\mathbf{x}(t) = [x_1(t) \hspace{2mm} x_2(t)]^T$, satisfies:
\begin{eqnarray}
x_1(t) = x_{1}^0-\int_0^t(u_{(1,2)}(\tau)x_1(\tau)-u_{(2,1)}(\tau)x_2(\tau))d\tau, \label{eq:ineq1} \\
x_2(t) = x_{2}^0 + \int_0^t(u_{(1,2)}(\tau)x_1(\tau)-u_{(2,1)}(\tau)x_2(\tau))d\tau, \label{eq:ineq2}
\end{eqnarray}
such that $x_{1}^0 \in (0,1)$ and $x_{2}^0 = 1-x_{1}^0$.
We assume, without loss of generality, that $x_1(t)>0$ for all $t \in [0,1)$.  Then for each $T \in [0,1)$, Equations \eqref{eq:ineq1} and \eqref{eq:ineq2} imply that:
\begin{eqnarray}
x_1(T) \hspace{-2mm} &=& \hspace{-2mm} x^0_{1} -  \nonumber \\
 && \hspace{-2mm} \int_0^T \left(u_{(1,2)}(\tau)+ u_{(2,1)}(\tau) -\frac{u_{(2,1)}(\tau)}{x_1(\tau)}\right)x_1(\tau)d\tau. \nonumber
\end{eqnarray}
From this equation, we can conclude that
\begin{eqnarray}
\label{eq:ineq1b}
\hspace{-4mm} x_1(1)
&\geq& x_{1}^0 -\int_0^1 (u_{(1,2)}(\tau)+u_{(2,1)}(\tau)\tilde{x}_1(\tau))d\tau \\ \nonumber
&\hspace{2mm}& = \exp{\left(-\int_0^1(u_{(1,2)}(\tau)+u_{(2,1)}(\tau)) d\tau \right)} x_{1}^0,
\end{eqnarray}
where $\tilde{x}_1$ is the solution of the differential equation
\begin{eqnarray}
\label{eq:lstctr}
\dot{\tilde{x}}_1(t) &=& -(u_{(1,2)}(t)+u_{(2,1)}(t))\tilde{x}_1(t), \\
 \tilde{x}_1(0) &=& x_{1}^0. \nonumber
\end{eqnarray}
Therefore, it must be true that $\exp{(-\int_0^1(u_{(1,2)}(\tau)+u_{(2,1)}(\tau))d \tau}) x_{1}^0$ $\leq 0$, which yields a contradiction since ${x}_{1}^0 \neq 0$.
\end{example}

The above observation is not a significant disadvantage, since each point on the boundary of $\mathcal{P}(\mathcal{V})$ is at least {\it asymptotically controllable}, a result that we prove in \cite{biswal2016sos}.

\subsection{Stabilization}

Now we investigate the stabilizability properties of the system \eqref{eq:ctrsys}. Note that stabilizability using centralized feedback follows from the controllability result in Theorem \ref{MK0}. Hence, our focus in this section is to establish stabilizability using decentralized control laws.

\begin{lemma}
\label{linstab}
Let $\mathcal{G}$ be strongly connected, and define $\mathbf{x}^{eq} \in {\rm int}(\mathcal{P}(\mathcal{V}))$. For each $e \in \mathcal{E}$ and each $\mathbf{y} \in \mathbb{R}^M$, let $k_e:\mathbb{R}^M \rightarrow (-\infty, \infty)$ be given by $k_e(\mathbf{y})=  x^{eq}_{T(e)}y_{S(e)} -x^{eq}_{S(e)}y_{T(e)} $ in system $\eqref{eq:fctrsys}$. Then, $\mathbf{x}^{eq}$
is locally exponentially stable on the space $\mathcal{P}(\mathcal{V})$.  That is, there exists $r > 0$ such that $\| \mathbf{x}^0 - \mathbf{x}^{eq} \|_2 < r$ and $\mathbf{x}^0 \in\mathcal{P}(\mathcal{V})$ imply that the solution $\mathbf{x}(t)$ of system \eqref{eq:fctrsys} satisfies the following inequality,
\begin{equation}
\|\mathbf{x}(t) - \mathbf{x}^{eq}\|_2 \leq M_0 e^{-\lambda t},
\label{eq:expdec}
\end{equation}
for all $t \in [0, \infty)$ and for some parameters $M_0>0$ and $\lambda >0$ that depend only on $r$.
\end{lemma}
\begin{proof}
We use linearization to establish local exponential stability. Consider the vector field $\mathbf{f}^e= [f^e_1 ~ f^e_2~...~f^e_M]^T$ given by
\[
f^e_i (\mathbf{y}) =
  \begin{cases}
   -(x^{eq}_{T(e)}y_{S(e)}- x^{eq}_{S(e)}y_{T(e)}) y_{S(e)} & \text{if }  i = S(e),\\
    (x^{eq}_{T(e)}y_{S(e)}- x^{eq}_{S(e)}y_{T(e)}) y_{S(e)} & \text{if } i = T(e), \\
   0       & \text{otherwise}
  \end{cases}
\]
for each $\mathbf{y} \in \mathbb{R}^M$. Then for each $e \in \mathcal{E}$, we define the matrix $\mathbf{A}_{e} \in \mathbb{R}^M \times \mathbb{R}^M$ as follows:
\[
A_e^{ij} =
  \begin{cases}
\frac{\partial f^e_{S(e)}}{\partial y_{S(e)}} \Bigr|_{\mathbf{y}=\mathbf{x}^{eq}}  = -x^{eq}_{T(e)}x^{eq}_{S(e)} & \hspace{-1mm} \text{if } i = j  = S(e), \\
 \frac{\partial f^e_{S(e)}}{\partial y_{T(e)}} \Bigr|_{\mathbf{y} =\mathbf{x}^{eq}} = (x^{eq}_{S(e)})^2 & \hspace{-2mm} \text{if } i=S(e),j  =T(e), \\
\frac{\partial f^e_{T(e)}}{\partial y_{T(e)}} \Bigr|_{\mathbf{y} =\mathbf{x}^{eq}} = -(x^{eq}_{S(e)})^2 & \hspace{-2mm} \text{if }  i= j  =T(e), \\
\frac{\partial f^e_{T(e)}}{\partial y_{S(e)}} \Bigr|_{\mathbf{y}=\mathbf{x}^{eq}} = x^{eq}_{T(e)}x^{eq}_{S(e)} & \hspace{-2mm} \text{if }   i=T(e),j  = S(e), \\
 0 \hspace{2mm} & \hspace{-2mm} \text{otherwise. }
  \end{cases}
\]

Now we define the matrix $\mathbf{G} \in \mathbb{R}^{M \times M}$ as $\mathbf{G} = \sum_{e\in \mathcal{E}}\mathbf{A}_e$. Note that $G^{S(e)T(e)} > 0$ for each $e \in \mathcal{E}$, since $\mathbf{x}^{eq} \in {\rm int}(\mathcal{P}(\mathcal{V}))$. Moreover, $\mathbf{1}^T \mathbf{G} = \mathbf{0}$, and the off-diagonal terms of $\mathbf{G}$ are positive. Hence, $\mathbf{G}$ is an irreducible transition rate matrix.  It is a classical result that this implies that $\mathbf{G}$ has its principal eigenvalue at $0$, which is simple. The other eigenvalues of $\mathbf{G}$ lie in the open left-half of the complex plane. However, note that the equilibrium point $\mathbf{x}^{eq}$ is {\it non-hyperbolic}, since the principal eigenvalue of $\mathbf{G}$ is at $0$. Hence, local exponential stability of the nonlinear system does not immediately follow. However, it follows that there exists an $(M-1)-$dimensional local stable manifold of the system that is tangential to $\mathcal{P}(\mathcal{V})$ at $\mathbf{x}^{eq} \in \mathcal{P}(\mathcal{V})$. Noting that the set $\lbrace \mathbf{y} \in \mathbb{R}^M ; \sum_{i=1}^M y_i = c \rbrace$ is invariant for solutions of the system \eqref{eq:fctrsys} for any $c \in \mathbb{R}$, it follows that the stable manifold is in fact in $\mathcal{P}(\mathcal{V})$. From this, the result follows.
\end{proof}

The above lemma implies that if negative transition rates are admissible, then there exists a linear feedback law, $\lbrace k_e\rbrace _{e\in \mathcal{E}}$, such that $k_e(\mathbf{x}^{eq}) = 0$ for each $e \in \mathcal{E}$ and the desired equilibrium point is locally exponentially stable.

In the above lemma, we have only established local exponential stability. We can also show global stability using an appropriate Lyapunov function, as we do in \cite{biswal2016sos}. We exclude this proof due to space constraints.

A desirable property of the control system \eqref{eq:ctrsys} is that stabilization of the desired equilibrium can be achieved using a linear feedback law that satisfies positivity constraints away from equilibrium and is zero at equilibrium. However, any stabilizing linear control law that is zero at equilibrium must in fact be zero everywhere. On the other hand, in the next theorem we show that whenever $\mathcal{G}$ is bidirected, any feedback control law that violates positivity constraints can be implemented using a rational feedback law of the form $k(\mathbf{x}) = a(\mathbf{x})+ b(\mathbf{x})\frac{f(\mathbf{x})}{g(\mathbf{x})}$, such that $k(\mathbf{x})$ satisfies the positivity constraints and is zero at equilibrium.

\begin{theorem}
\label{ratio}
Let $\mathcal{G}$ be a bidirected graph. Let $k_e:\mathbb{R}^M \rightarrow (-\infty, \infty)$ be a map for each $e \in \mathcal{E}$ such that there exists a unique global solution of the system \eqref{eq:fctrsys}. Additionally, assume that $\mathbf{x}(t) \in {\rm int}(\mathcal{P}(\mathcal{V}))$ for each $t\in [0,\infty)$. Consider the functions $m_e^p:\mathbb{R}^M\rightarrow \lbrace 0,1 \rbrace$ and $m^n_e:\mathbb{R}^M \rightarrow \lbrace -1,0 \rbrace$, defined as follows for each $e \in \mathcal{E}$:
\begin{eqnarray}
m^p_e(\mathbf{y}) &=& ~1 \hspace{2mm} \mbox{ if } \hspace{2mm} k_e(\mathbf{y}) \geq 0, 
 ~~~0 \hspace{1mm} \mbox{ otherwise;} \nonumber \\
m^n_e(\mathbf{y}) &=& \hspace{-1mm} -1 \hspace{2mm} \mbox{ if } \hspace{2mm} k_e(\mathbf{y}) \leq 0, 
 ~~~0 \hspace{1mm} \mbox{ otherwise.}
\end{eqnarray}
Let $c_e:\mathbb{R}^M \rightarrow [0,\infty)$ be given by
\begin{equation}
c_e(\mathbf{y}) = m^p_e(\mathbf{y})k_e(\mathbf{y}) - m^n_{\tilde{e}}(\mathbf{y}) k_{\tilde{e}}(\mathbf{y})\frac{y_{S(e)}}{y_{T(e)}}.
\end{equation}
Then the solution $\tilde{\mathbf{x}}(t)$ of the following system,
\begin{eqnarray}
\dot{\tilde{\mathbf{x}}}&=& \sum_{ e\in \mathcal E} c_e(\tilde{\mathbf{x}}(t)) \mathbf{B}_e \tilde{\mathbf{x}}(t), \hspace{3mm} t \in [0, \infty), \nonumber
 \\ \tilde{\mathbf{x}}(0) &=& \mathbf{x}^0 \in \mathcal{P}(\mathcal{V}), \label{eq:ctrsysmod}
\end{eqnarray}
is unique, defined globally, and satisfies $\tilde{\mathbf{x}}(t) = \mathbf{x}(t)$ for all $t \in [0,\infty)$.
\end{theorem}
\begin{proof}
This follows by noting that the right-hand sides of systems \eqref{eq:ctrsysmod} and \eqref{eq:fctrsys} are equal for all $t\geq 0$.
\end{proof}

\begin{remark}
In the above theorem, it is required that $\tilde{\mathbf{x}}(t) \in {\rm int}(\mathcal{P}(\mathcal{V}))$ for all $t \in [0,\infty)$. Such an assumption on the initial distribution $\mathbf{x}^0$ can be avoided if one uses polynomial feedback instead, as we show in \cite{biswal2016sos}.
\end{remark}

The above theorem can also be extended to the case of strongly connected graphs, since the cone spanned by the collection of vectors $\lbrace \mathbf{B}_e \mathbf{x} \rbrace_{e \in\mathcal{E}}$ is equal to the tangent space of $\mathcal{P}(\mathcal{V})$ at $\mathbf{x}$ whenever $\mathbf{x} \in {\rm int}(\mathcal{P}(\mathcal{V}))$.  Hence, negative flows can be realized using positive inputs of appropriate magnitude, possibly applied across more than one edge, depending on the length of the shortest directed path in the opposite direction. However, the resulting control law might not be decentralized in the sense that it might not respect the graph structure.

\section{COMPUTATIONAL METHOD FOR CONTROLLER DESIGN} 
\label{sec:comp}
In this section, we discuss how decentralized linear control laws can be algorithmically constructed for the system \eqref{eq:ctrsys}. These control laws always violate the positivity constraint. Hence, in order to implement them, we can design a rational feedback law that respects the positivity constraints but has the same effect as the corresponding linear feedback, as shown by Theorem \ref{ratio}.
The following result is well-known \cite{dullerud2013course} in literature on Linear Matrix Inequality (LMI) based tools for construction of linear control laws:
\begin{theorem}
Let $\mathbf{A} \in \mathbb{R}^{M \times M}$ and $\mathbf{B} \in \mathbb{R}^{M \times N_{\mathcal{E}}}$, where  $N_{\mathcal{E}}$ is the number of control inputs. Consider the linear control system
\begin{eqnarray}
\dot{\mathbf{x}}(t) &=& \mathbf{A}\mathbf{x}(t)+\mathbf{B}\mathbf{u}(t). \label{eq:ABsys}
\label{eq:Linsys}
\end{eqnarray}
Then a static linear state feedback law, $\mathbf{u} = -\mathbf{K}\mathbf{x}$, stabilizes the system \eqref{eq:Linsys} if and only if there exist matrices $\mathbf{P}>0$ and $\mathbf{Z}$ such that
\begin{equation}
\mathbf{K} = \mathbf{Z}\mathbf{P}^{-1},
\end{equation}
\begin{equation}
\mathbf{P}\mathbf{A} + \mathbf{B}\mathbf{Z} + \mathbf{P}\mathbf{A}^T +\mathbf{Z}^T\mathbf{B}^T < 0.
\label{eq:mnlmi}
\end{equation}
\end{theorem}

The above theorem can used to construct state feedback laws for a general linear system. The theorem is attractive from a computational point of view since the constraints, $\mathbf{P}>0$ and Equation \eqref{eq:mnlmi}, are convex in the decision variables $\mathbf{Z}$ and $\mathbf{P}$. In order to ensure that the resulting control law is decentralized, we need to impose additional constraints. Toward this end, let $\mathcal{D} \subset \mathbb{R}^{M \times M}$ be the subset of diagonal positive definite matrices. Additionally, let $\mathcal{Z} = \lbrace \mathbf{Y} \in \mathbb{R}^{N_\mathcal{E} \times M} : Y^{ij} = 0  \hspace{1mm}$ if $i \neq j$ and  $(i,j) \in \mathcal{V} \times \mathcal{V}- \mathcal{E}  \rbrace$.  Then the additional constraints $\mathbf{Z} \in \mathcal{Z}$ and $\mathbf{P} \in \mathcal{D}$
can be imposed to achieve the desired decentralized structure in the controller. Note that these two constraints are convex and can be expressed as LMIs. Hence, state-of-the-art LMI solvers can be used to construct control laws. It follows from Lemma \ref{linstab} that the constraints $\mathbf{Z} \in \mathcal{Z}$ and $ \mathbf{P} \in \mathcal{D}$ are always feasible for linearizations of the system \eqref{eq:ctrsys}. 
The existence of a quadratic Lyapunov function can also be inferred from this lemma, since the linearization of the closed-loop nonlinear system in Lemma \ref{linstab} is always the forward equation of an irreducible CTMC.

The above method of constructing decentralized control laws is well-known in the literature. However, a necessary and sufficient condition for the above LMIs to be feasible is that the system $(\mathbf{A},\mathbf{B})$ (Equation \eqref{eq:ABsys}) is stabilizable. This system is not stabilizable on $\mathbb{R}^{M}$, since the uncontrollable eigenvalue is at zero and the set $\mathcal{P}(\mathcal{V})$ is invariant for system \eqref{eq:ctrsys}. However, we require stability only on $\mathcal{P}(\mathcal{V})$. To deal with the lack of stabilizability on $\mathbb{R}^M$, an alternative approach from model reduction is to artificially place the uncontrollable eigenvalue of the linear system in the open left half of the complex plane. To see this explicitly, let $(\mathbf{A},\mathbf{B})$ be a partially controllable system. Then there exists a nonsingular matrix $\mathbf{T} \in \mathbb{R}^{M \times M}$ such that
\begin{equation}
\tilde{\mathbf{A}} = \mathbf{T}\mathbf{A}\mathbf{T}^{-1} =
\begin{bmatrix}
    \tilde{\mathbf{A}}_{11} & \tilde{\mathbf{A}}_{12} \\
    \mathbf{0}              & \tilde{\mathbf{A}}_{22}
\end{bmatrix}, \hspace{2mm}
\tilde{\mathbf{B}} = \mathbf{T}\mathbf{B} =
\begin{bmatrix}
    \tilde{\mathbf{B}}\\
    \mathbf{0}
\end{bmatrix}.
\end{equation}
In our case, $\tilde{\mathbf{A}}=\mathbf{A}=\mathbf{0}$. However, this transformation is needed to convert $\tilde{\mathbf{B}}$ into the desired form. In order to design a controller for the system $(\mathbf{A},\mathbf{B})$, we can instead design a controller for another artificial system,
\begin{equation}
\tilde{\mathbf{A}}_{\epsilon} = \mathbf{T}\mathbf{A}\mathbf{T}^{-1} =
\begin{bmatrix}
    \tilde{\mathbf{A}}_{11} & \tilde{\mathbf{A}}_{12} \\
    \mathbf{0}                & -\epsilon \mathbf{I}
\end{bmatrix}, \hspace{2mm}
\tilde{\mathbf{B}} = \mathbf{T}\mathbf{B} =
\begin{bmatrix}
    \tilde{\mathbf{B}}\\
    \mathbf{0}
\end{bmatrix}
\end{equation}
for some $\epsilon > 0$, where $\mathbf{I}$ is the identity matrix of appropriate dimension. Note that the new artificial system, $(\tilde{\mathbf{A}}_{\epsilon},\tilde{\mathbf{B}})$, has the same controllable eigenvalues as the original system, $(\tilde{\mathbf{A}},\tilde{\mathbf{B}})$, and all its uncontrollable eigenvalues are stable. We can perform the inverse transformation to represent the artificial system in the original coordinates:
\begin{equation}
\mathbf{A}_{\epsilon}= \mathbf{T}^{-1}\tilde{\mathbf{A}}_{\epsilon}\mathbf{T}.
\end{equation}
Now let $\mathbf{F}$ be a feedback control law that stabilizes the system $(\mathbf{A}_{\epsilon},\mathbf{B})$.  Then the following relation is satisfied, where $\sigma(\mathbf{M})$ is the spectrum of matrix $\mathbf{M}$:
\begin{equation}
\sigma(\mathbf{A}+\mathbf{B}\mathbf{F}) \backslash \sigma(\tilde{\mathbf{A}}_{22}) = \sigma(\mathbf{A}_{\epsilon}+\mathbf{B}\mathbf{F}) \backslash \lbrace \epsilon \rbrace.  \nonumber
\end{equation}
This relation is advantageous from a computational point of view, since one can now directly impose the constraint of decentralized structure on the gain matrix $\mathbf{F}$ by designing the controller for the stabilizable artificial system $(\mathbf{A}_{\epsilon},\mathbf{B})$ and then implementing it on the original system.
Alternatively, if we were to first reduce the original system to a controllable lower-dimensional system and then design the controller for this reduced-order system, the structural constraint would be harder to impose.

Combining all the LMIs described, we obtain the following system of LMIs that need to be tested for feasibility:
\begin{eqnarray}
\mathbf{P} \in \mathcal{D}, ~~ \mathbf{Z} \in \mathcal{Z}, ~~ \mathbf{P}>0,  \\ \nonumber
\mathbf{P}\mathbf{A}_{\epsilon} + \mathbf{B}\mathbf{Z} + \mathbf{P}\mathbf{A}_{\epsilon}^T +\mathbf{Z}^T\mathbf{B}^T<0.
\end{eqnarray}

\section{NUMERICAL SIMULATIONS}

In this section, we numerically verify the effectiveness of decentralized feedback controllers that we compute for two bidirected graphs.  For both graphs, the numerical solution of the mean-field model \eqref{eq:fctrsys} was compared to stochastic simulations of the CTMC characterized by expression \eqref{eq:marko}. This CTMC was simulated using an approximating DTMC that evolves in discrete time, where the probability that an agent $i$ at vertex $S(e)$, $e \in \mathcal{E}$, at time $t$ transitions to vertex $T(e)$ at time $t+\Delta t$ was set to:
\begin{align*}
\mathbb{P}(\tilde{X}_i&(t+\Delta t) = T(e) | \tilde{X}_i(t) = S(e))  \nonumber \\
\hspace{-5mm} =& ~ k_e(\tfrac{1}{N}\mathbf{N}(t)) \Delta t
 ~=~ c_e\left(\tfrac{1}{N}N_{S(e)}(t),\tfrac{1}{N}N_{T(e)}(t)\right)\Delta t. \nonumber
\end{align*}
Here, $c_e:\mathbb{R}^2 \rightarrow [0,\infty)$ is the feedback law for agents transitioning along edge $e$, and $\mathbf{N}(t) = [N_1(t)~N_2(t)~...~N_M(t)]^T$ is the vector of agent populations at each vertex at time $t$. The control law $k_e$ is decentralized in the sense that the probability rate of an agent performing the transition associated with edge $e$ depends only on the agent populations at the corresponding source vertex $S(e)$ and target vertex $T(e)$. Here, we are assuming that each agent can measure the agent populations at its current vertex and adjacent vertices.






\subsection{Four-vertex graph}
We computed three different types of controllers to redistribute populations of $N=50$ and $N=500$ agents on the four-vertex chain graph in Fig. \ref{fig:fournode}.  The first controller ({\it Case 1}) was a reference open-loop time-invariant controller; the second ({\it Case 2}) was a feedback controller constructed as described in Lemma \ref{linstab}; and the third controller ({\it Case 3}) was designed using the LMI-based computational approach discussed in Section \ref{sec:comp}. The controller for {\it Case 3} was constructed by setting the right-hand side of system \eqref{eq:ctrsys} equal to $\mathbf{G}\mathbf{x} = -\mathbf{L(\mathcal{G})}\mathbf{D}\mathbf{x}$, where $\mathbf{L(\mathcal{G})}$ is the Laplacian matrix of the graph $\mathcal{G}$ and $\mathbf{D}$ is a diagonal matrix with entries $D^{ij} = 1/x^{eq}_i$ if $i=j$, $D^{ij} = 0$ otherwise. This makes the desired distribution $\mathbf{x}^{eq}$ invariant for the corresponding CTMC. The transition rates (control inputs) for the {\it Case 3} controller were defined as $u_e(t)=G^{T(e)S(e)}$ for all $t \in [0,\infty)$, $e \in \mathcal{E}$.
In all three cases, the initial distribution was $\mathbf{x}^0 = \tfrac{1}{N}\mathbf{N}(0) = [0.7~ 0.1~ 0.1~ 0.1]^T $, and the desired distribution was $\mathbf{x}^{eq} = [0.1~ 0.1~ 0.1~ 0.7]^T$.

\begin{figure}[t!]
	\centering
	\includegraphics[width= \linewidth]{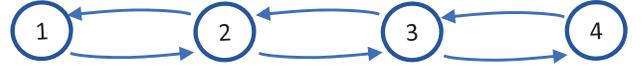}
	\caption{Four-vertex graph with one-dimensional grid structure.}
	\label{fig:fournode}
\end{figure}


The solution of the mean-field model and trajectories of the stochastic simulation are compared in Fig. \ref{fig:openloop}-\ref{fig:lmi500}.  Fig. \ref{fig:openloop} shows that the open-loop controller ({\it Case 1}) produces large variances in the agent populations at the vertices at steady-state.  As an expected consequence of the law of large numbers, these variances are smaller for $N=500$ agents than for $N=50$ agents. In comparison, the variances are much smaller when the {\it Case 2} and {\it Case 3} feedback controllers are used, as shown in Fig. \ref{fig:prop50} and \ref{fig:lmi500}. This is due to the property of the feedback controllers that as the agent densities approach their desired equilibrium values, the transition rates tend to zero. This property reduces the number of unnecessary agent state transitions at equilibrium. 
This effect is shown explicitly in Fig. \ref{fig:robotmotion}, which plots the time evolution of a single agent's state (vertex number) during a stochastic simulation with each of the three controllers. As expected, using open-loop control, the agent's state keeps switching between  vertices and never reaches a steady-state value. In contrast, using the feedback controllers, the agent's state remains constant after a certain time. 



\begin{figure}[t]
	\centering	
	\subfigure[$N=50$ agents \label{fig:openloop50}]{\includegraphics[width=\linewidth]{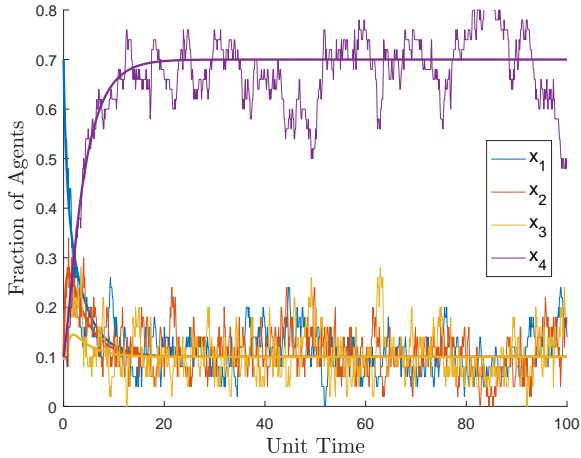}}
     \subfigure[$N=500$ agents \label{fig:openloop500}]{\includegraphics[width=\linewidth]{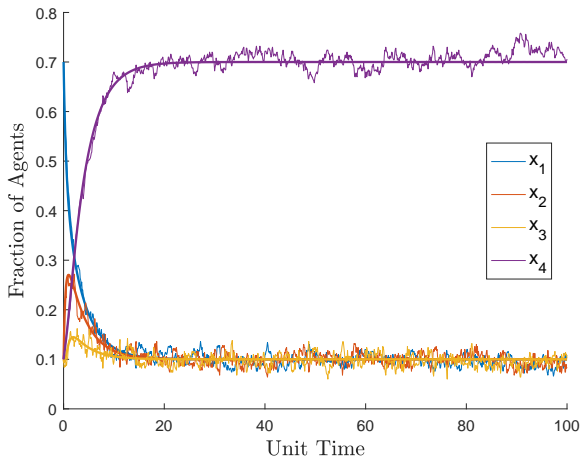}}			
    \caption{Trajectories of the mean-field model \eqref{eq:fctrsys} {\it (thick lines)} and the corresponding stochastic simulation {\it (thin lines)} for the {\it Case 1} open-loop controller.}
	\label{fig:openloop}
\end{figure}

%

\begin{figure}[t]
	\centering
	\includegraphics[width= \linewidth]{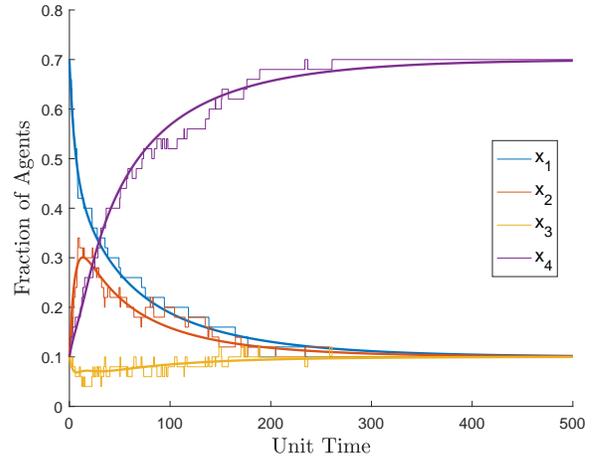}
	\caption{Trajectories of the mean-field model \eqref{eq:fctrsys} {\it (thick lines)} and the corresponding stochastic simulation {\it (thin lines)} for the {\it Case 2} controller from Lemma  \ref{linstab} with $N=50$ agents.}
	\label{fig:prop50}
\end{figure}



\begin{figure}[t]
	\centering
	\includegraphics[width= \linewidth]{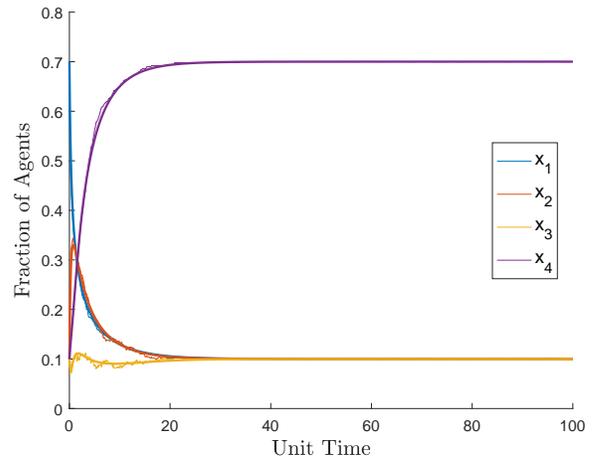}
	\caption{Trajectories of the mean-field model \eqref{eq:fctrsys} {\it (thick lines)} and the corresponding stochastic simulation {\it (thin lines)} for the {\it Case 3} LMI-based controller with $N=500$ agents.}
	\label{fig:lmi500}
\end{figure}

\begin{figure}[t]
	\centering
	\includegraphics[width= \linewidth]{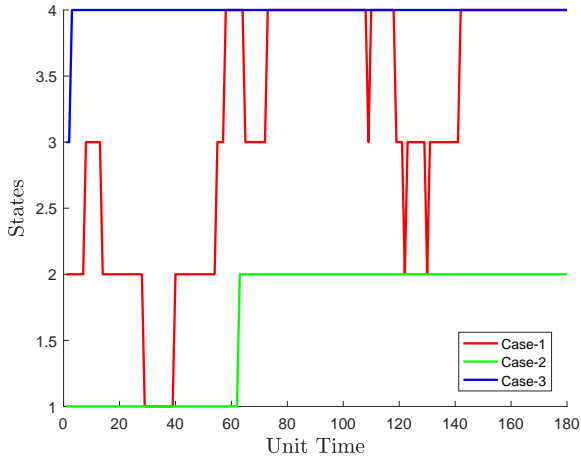}
	\caption{State (vertex number) of a single agent over time during stochastic simulations of system \eqref{eq:fctrsys} with the {\it Case 1}, {\it Case 2}, and {\it Case 3} controllers.}
	\label{fig:robotmotion}
\end{figure}

\subsection{Nine-vertex graph}
%

Here we demonstrate that our LMI-based design of decentralized control laws can be scaled with the number of vertices and control inputs.   
We computed an LMI-based controller to redistribute $N=500$ agents on the nine-vertex graph in Fig. \ref{fig:ninenode} and ran a stochastic simulation of the resulting control system. Fig. \ref{fig:sim1} illustrates the initial distribution of the agents on a two-dimensional domain, in which each partitioned region corresponds to a vertex in Fig. \ref{fig:ninenode}.  We assume that agents switch between regions instantaneously once they decide to execute a transition.  Fig. \ref{fig:sim2} shows the distribution of the agents at $t=300$, which matches the desired equilibrium distribution.




We note that the controller designed using LMIs is not guaranteed to be globally stabilizing, and might also not keep which implies that the set $\mathcal{P}(\mathcal{V})$ is not necessarily invariant for the system. As a result, we found that the LMI-based controller sometimes did not stabilize the mean-field model if the initial agent distribution was too far away from the equilibrium distribution.  On the other hand, the controller from Lemma \ref{linstab} was not subject to this limitation. 
Hence, additional constraints might need to be imposed on the LMI-based controller design to guarantee global stability.


%

\begin{figure}[t]
\label{fig:graph2}
	\centering
	\includegraphics[scale = 0.5]{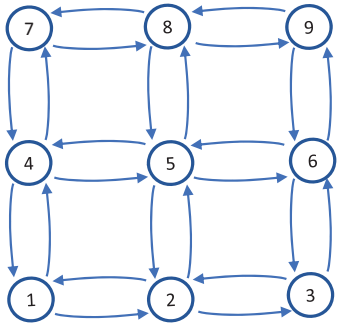}
	\caption{Nine-vertex graph with two-dimensional grid structure.} 
	\label{fig:ninenode}
\end{figure}

\begin{figure}[t]
	\centering	
	\subfigure[$t=0$ \label{fig:sim1}]{\includegraphics[width=0.65\linewidth]{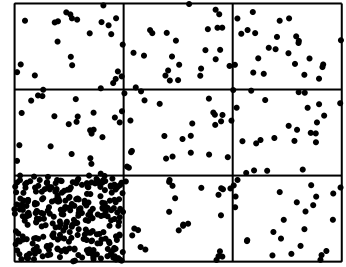}}
     \subfigure[$t=300$ \label{fig:sim2}]{\includegraphics[width=0.65\linewidth]{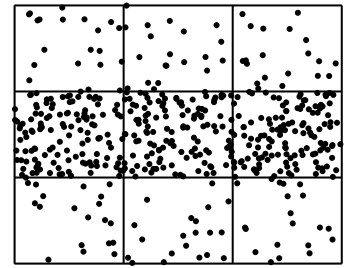}}			
    \caption{Snapshots of agent positions during a stochastic simulation for the nine-vertex graph. The initial and desired agent distributions among the nine regions (vertices) are $\mathbf{x}^0 = [0.6~~ 0.05  ~~ 0.05  ~~ 0.05  ~~ 0.05  ~~ 0.05  ~~ 0.05  ~~ 0.05  ~~ 0.05  ~~ 0.05]^T$ and $\mathbf{x}^{eq} = [0.04 ~~ 0.04 ~~ 0.04 ~~0.25 ~~0.25 ~~ 0.26 ~~0.04 ~~0.04 ~~ 0.04]^T$.}
	\label{fig:sims}
\end{figure}

\section{CONCLUSION}

In this paper, we have proven local and global controllability properties of the forward equation of CTMCs. Moreover, we presented novel approaches to designing decentralized linear feedback control laws for a robotic swarm whose distribution among a set of tasks, represented by a graph, is governed by this equation. Since linear feedback laws violate positivity constraints, it was shown that for bidirected graphs, linear feedback laws can be realized using rational feedbacks that have the same decentralized structure. These feedback laws ensure that agents do not unnecessarily switch between tasks once the desired equilibrium distribution is reached.

Future work will focus on questions regarding the optimality of these control laws and the possibility of stabilization for general strongly connected graphs.  In a companion paper \cite{biswal2016sos}, we prove asymptotic controllability of points on the boundary of the simplex $\mathcal{P}(\mathcal{V})$, and we plan to further investigate stabilization to these points.  As we also show in \cite{biswal2016sos}, for bidirected graphs, interior points of the simplex can be globally stabilized using decentralized quadratic feedback laws.  This property enables the use of sum-of-squares based computational polynomial optimization methods to construct control laws. 


\bibliographystyle{plain}

\bibliography{cdcref}

\end{document}